\documentclass[twoside]{article}

\usepackage[utf8]{inputenc}
\usepackage[english]{babel}
\usepackage{amsthm}
\usepackage{amsmath}
\usepackage{amssymb}

\theoremstyle{plain}
\newtheorem{theorem}{Theorem}

\theoremstyle{definition}

\date{}
\begin{document}

\title{On a necessary condition for the matching cryptosystem stability}
\author{Aleksey I. Bolotnikov and Anwar A. Irmatov
}
\maketitle
\begin{abstract}
    The article contains a description of a possible attack on a matching cryptosystem and a defense with limited noise. A public key of a matching cryposystem consists of a graph and a weight vector-function on the edges of the graph with values from a finite field, where a private key contains another weight function, for which the corresponding alternating weighted path problem can be solved in polynomial time. There is a specific family of these secret weight functions that is considered in this article, for which some of the coordinates of its vector values are described as limited noise. We suggest a necessary condition for the matching cryptosystem stability in terms of dimensions of spans of weight vectors that correspond to specific sets of edges of the graph from the public key.
\end{abstract}
The representation of the maximum matching problem in a graph as a linear programming problem on a matching
polyhedron, introduced by J. Edmonds in 1965 [1], is well known. Using the one-to-one correspondence between the matchings
of a graph $G$ and the vertices of the matching polyhedron $P(G)$, one can define the weight of a vertex of $P(G)$ as the weight of
the corresponding matching from $G$. Then a weight function on the edges of the graph with values in a finite-dimensional space over a finite field $\mathbb{F}_p$, where $p$ is a prime, can be characterized as "proper" or "improper" using the following property of the matching
polyhedron. Two vertices of $P(G)$ are connected by an edge if and only if the symmetric difference of the matchings corresponding
to these vertices consists of exactly one simple path of arbitrary length or one simple cycle of even length. Then we call the original weight function $W$
on the edges of $G$ "proper" if, for the constructed weight function on the vertices of $P(G)$, it is true that any two vertices of $P(G)$
connected by an edge have different weights. We call $W$ "improper" in the opposite case, that is, when there are at least two adjacent vertices
in $P(G)$ with the same weight. In this case, a simple path or a simple cycle of even length  $(e_{i_1}, ..., e_{i_d})$, such that  $W(e_{i_1}) - W(e_{i_2}) +W(e_{i_3}) + ... + (-1)^{d-1}W(e_{i_d}) = 0 \in \mathbb{F}^{k}_{p}$,
will be the symmetric difference of two matchings corresponding to vertices of $P(G)$ with the same weight. This allows us to
formulate the alternating weighted path problem.

{\bfseries Definition 1.} The alternating weighted path problem: given a graph $G$, a prime number $p$, a positive integer number $k$, a weight function $W : E \rightarrow \mathbb{F}^{k}_{p}$ and a vector $(a_1, ..., a_k) \in \mathbb{F}^{k}_{p}$, the problem is to determine, if there is a sequence of edges $(e_{i_1} , ..., e_{i_d})$ in graph $G$ that forms either a simple path of arbitrary length or a simple cycle of even length, such that $W(e_{i_1}) - W(e_{i_2}) + ... + (-1)^{d+1}W(e_{i_d}) = (a_1, ..., a_k).$

In [2], the alternating weighted path problem was proven to be NP-complete. The authors of [2] proposed a new knapsack-type public-key encryption system in which the NP-completeness of the alternating weighted path problem is
used to prove its security. 

{\bfseries Definition 2.} A matching cryptosystem is an encryption system with the following scheme:

\begin{itemize}
    \item {\bfseries Public Key}: Graph $G(V,E)$, prime number $p$, natural number $k$, and weight function $W:E\rightarrow \mathbb{F}_{p}^{k}$.
    \item {\bfseries Private Key}: A weight function $\hat W$, for which the corresponding alternating weighted path problem can be solved in polynomial time, and an invertible linear operator $M: \mathbb{F}_{p}^{k} \rightarrow \mathbb{F}_{p}^{k}$ that transforms weights on edges, such that $W$ is transformed into $\hat W$. If any additional objects were used in the construction of $\hat W$, they are considered a part of the private key as well.
    \item {\bfseries Message Space}: Simple paths and simple even cycles of $G$.
    \item {\bfseries Encryption}: For a path or an even cycle $(e_{i_1},...,e_{i_d})$, compute $(a_1,...,a_k) = W(e_{i_1})-W(e_{i_2})+...+(-1)^{d+1}W(e_{i_d})$.
    \item {\bfseries Decryption}: Apply operator $M$ to the received message
    $A = (a_1,...,a_k)$ and solve the easy alternating weighted path problem
    with weight function $\hat W(e) = M(W(e))$ and the alternating sum of weights $B = M(A)$.
\end{itemize}

Consider the following family of quickly solvable alternating weighted path problems. Let $G$ be a complete graph $K_m$. Let $p$ be a big enough prime number, $p > 2 \cdot 3^{m-2}$. Let the dimension of the space of weights be 
$k = m$.

Function $\hat W$ is defined as follows. We fix a numbering of vertices of the graph $G: V = \{v_1, ..., v_m\}$. For an arbitrary pair of vertices $v_i, v_j , i < j$, let the weight of the edge  $(v_i, v_j)$ be denoted as $\hat W_{i,j} = (\hat w_{i,j,1}, ...,\hat w_{i,j,k}).$ For each $1 \leq i \leq m - 1$ the sequence $C_i = (\hat w_{i,i+1,i},\hat w_{i,i+2,i}, ...,\hat w_{i,m,i})$ is a rapidly growing sequence:  $3\cdot \hat w_{i,j,i} \leq \hat w_{i,j+1,i}.$ In total, $m(m - 1)/2$ of  elements $\hat w_{i,j,t}$ are used as elements of rapidly growing sequences. Other elements $\hat w_{i,j,t}$ are used as noise to improve the system's resistance to attacks. Since there is $m^{2}(m-1)/2,$ of elements $\hat w_{i,j,t}$ in general and $m(m-1)/2$ of them used in rapidly growing sequences, $m(m - 1)^{2}/2$ of elements 
$\hat w_{i,j,t}$ are used as noise.

An algorithm that solves problems from this family, when  $C_i = (1, 3,..., 3^{m-i-1})$, $\hat w_{i,j,t}= 0, t \neq i, t \neq j$, and  $\hat w_{i,j,t}= 1, t = j$,  in polynomial time, is described in [2]. The algorithm works as follows. If a path of arbitrary length or a cycle of even length $(e_{i_1} , ..., e_{i_d})$, that satisfies
$\hat W(e_{i_1}) - \hat W(e_{i_2}) + ... + (-1)^{d+1}\hat W(e_{i_d}) = (a_1, ..., a_k),$  passes through the first vertex,
then $a_1$ from $(a_1, ..., a_k)$ can take one of the following values:
\begin{itemize}
\item $3^i, 0 \leq i \leq m - 2,$
\item $p - 3^i, 0 \leq i \leq m - 2,$
\item $3^j - 3^i, 0 \leq i < j \leq m - 2,$
\item  $p - (3^j - 3^i), 0 \leq i < j \leq m - 2.$
\end{itemize}
There are $m(m-1)$ of these values in total. And if $(e_{i_1} , ..., e_{i_d})$
does not pass through the first vertex, $a_1 = 0.$

By determining, if $a_1$ is equal to any of those $m(m - 1) + 1$ values, we can discover edges on $(e_{i_1} , ..., e_{i_d})$ that are adjacent to the first vertex. After weights of discovered edges are subtracted from $(a_1, ..., a_k)$ with correct signs, the procedure is repeated for the second vertex, then for the third vertex, etc.

Elements $\hat w_{i,j,t}$ of the weight function can be divided into three categories. Elements  $\hat w_{i,j,t}$ with $i = t$ are used in rapidly increasing sequences. Elements $\hat w_{i,j,t}$ with $ i < t$
have the following property: values of these elements do not affect the 
functioning of the algorithm described above. These elements will be referred to as arbitrary noise. Elements  $\hat w_{i,j,t}$ with $i > t$, will be reffered to as limited noise, since aribtrary values of these elements can affect the ability of the algorithm to find a distinct solution in polynomial time.

The following is an example of an alternating weighted path problem with a non-zero limited noise that can be solved in polynomial time. Instead of
 $C_j = (\hat w_{j,j+1,j} , ...,\hat w_{j,m,j}) = (1, 3, 3^2, ..., 3^{m-j-1})$, let $ C_j = (a_{j,1}, ..., a_{j,m-j}),$ where $\forall i \; 3\cdot a_{j,i-1} \leq a_{j,i}$ and $a_{j,1} \geq 3,$ and let $\hat w_{i,j,t} = 1, t < i$.
Such a modification makes it necessary for each $a_i$ to additionally check if $a_i$ is equal to any of $\pm(a_{j,i}-1)$ or $\pm (a_{j,t}-a_{j,i}) \pm 1$, and the total number of value comparisons stays polynomial.

The necessity of using non-zero limited noise is due to the following result.

 \begin{theorem} There is an attack that operates in time polynomial to the public key size,
to which all systems  with zero limited noise from the described family are vulnerable.
\end{theorem}

\begin{proof}

Assume that we know the complete graph $G(V,E) = K_m$, the prime number $p > 2 \cdot max_{1 \leq i <j \leq m}(\hat w_{i,j,i})$, the function $W :E \rightarrow \mathbb{F}^{k}_{p}$ with the dimension of weight vectors $k = m$, and the transmitted message  $A_1 = (a_1,...,a_k)$. Assume that we do not know the numbering of vertices that was used in the definition of $\hat W$, and this numbering is not connected to the numbering of edges of $G$ that is used in the description of public key. Let $\hat W$ have zero limited noise.

The first iteration of the attack consists of the following steps:

\begin{enumerate}
\item Find all vertices $v$ with a property that the dimension of the span of the set  $F(v) = \{W(e)|e \in
E, e$ is not incident to  $v\}$ is less then $k$.
\item Keep those vertices $v$ that satisfy all the following conditions:

\begin{itemize}
\item for every edge $e$ that is incident to $v$, $W(e)$ does not belong to the span of $F(v)$;
\item  for every two distinct edges $e_i$ and $e_j$ that are incident to $v$, $W(e_i) -W(e_j)$ does not belong to the span of $F(v)$;
\item for every three distinct edges $e_i$, $e_j$ and $e_k$ that are incident to $v$, $W(e_i) - (W(e_j) - W(e_k))$ does not belong to the span of $F(v)$;
\item for every four distinct edges $e_i$, $e_j$,$e_k$ and $e_t$ that are incident to $v$, $(W(e_i) - W(e_j))  - (W(e_k) - W(e_t))$ does not belong to the span of $F(v)$.
\end{itemize}

\item Check for the transmitted message $A_1$ and each remaining vertex $v$, that
$A_1$ belongs to the span of $F(v)$. If it does not belong, then discover either an edge $e$ that is incident to $v$, such that $A_1 \pm W(e)$ belongs to the span of $F(v)$, or a couple of edges $e_i$ and $e_j$ that are incident to $v$, such that $A_1 \pm (W(e_i)-W(e_j))$
belongs to the span $F(v)$. If $A_1$ belongs to span of $F(v)$, we conclude, that the initial message does not pass through the vertex $v$.
\end{enumerate}

Let $u$ be the vertex of the graph that was the first vertex in the vertex numbering that was used in the construction of $\hat W$. We will show that $u$ satisfies conditions 1 an 2. 

Since $\hat W$ has zero limited noise, for every edge $e$ that is not incident to $u$, the first coordinate of the vector $\hat W(e)$ is zero. Therefore, the dimension of the span of the set $\hat F(u) = \{\hat W(e)|e \in
E, e$ is not incident to $v\}$ is less then $k$. Since $\hat W$ is obtained from $W$ by applying an invertible linear operator, $span(F(u))$ and $span(\hat F(u))$ have the same dimension. Change of numbering on edges of $G$ during the description of public key does not affect the dimension of $span(F(u))$, therefore 
$u$ satisfies condition 1, and we will find at least one vertex after step 1.

First coordinates of values of vector-function $\hat W$ on edges, that are incident to $u$, are not equal to zero and form a rapidly increasing sequence. Therefore, for every edge $e$ that is incident to $u$, $\hat W(e)$ does not belong to $span(\hat F(u))$, and for every two distinct edges $e_i$ and $e_j$ that are incident to $u$,  $\hat W(e_i) - \hat W(e_j)$ does not belong to $span(\hat F(u))$.

Now we will show that for every three distinct edges $e_i$, $e_j$, $e_k$ that are incident to $u$, $\hat W(e_i) - (\hat W(e_j)  - \hat W(e_k))$ does not belong to $span(\hat F(u))$. 
Let the first coordinate of $\hat W(e_i)$ be denoted as $\alpha$, let the first coordinate of  $\hat W(e_j)$ be denoted as  $\beta$, and let the first coordinate of  $\hat W(e_k)$ be denoted $\gamma$.
The first coordinate of $\hat W(e_i) - (\hat W(e_j)  - \hat W(e_k))$ is equal to $\alpha - (\beta -\gamma)  =  \alpha - \beta + \gamma $. Assume that  $ \alpha - \beta + \gamma \geq 0$. We will show that in this case the greatest of these three values is either  $\gamma$ or $\alpha$. Assume that this is not the case. Let $\beta$ be the greatest of these three values. Then $\beta \geq 3 \alpha$ and $\beta \geq 3\gamma$. Therefore, 
 $2\beta \geq 3(\gamma + \alpha)$, and so $\beta \geq 1.5(\gamma + \alpha) > \gamma + \alpha$, and  $  \alpha - \beta + \gamma < 0$, which is a contradiction to the original assumption. Next, without any loss in generality, let $\alpha$ be the greatest of these three values.
Since $  \alpha \geq  3\beta$, $\alpha - \beta + \gamma > \alpha - \beta  >0$, 
and since $p> 2\cdot \alpha$, and $\alpha \geq 3\gamma$, $\alpha - \beta + \gamma <  \alpha + \beta + \gamma < 2\alpha < p$. Therefore,
the first coordinate of $\hat W(e_i) - (\hat W(e_j)  - \hat W(e_k))$ is not equal to  $0\text{ (mod p)}$, and $\hat W(e_i) - (\hat W(e_j)  - \hat W(e_k))$ does not belong to $span(\hat F(u))$. 
The case, where $ \alpha - \beta + \gamma \leq 0$, is resolved in a similar way.

Now we will show, that for every four distinct edges $e_i$, $e_j$, $e_k$ and $e_t$ that are incident to $u$, $(\hat W(e_i) - \hat W(e_j))  - (\hat W(e_k) - \hat W(e_t))$ does not belong to $span(\hat F(u))$. 
Let the first coordinate of $\hat W(e_i)$ be denoted as $\alpha$,  let the first coordinate of $\hat W(e_j)$ be denoted as $\beta$, let the first coordinate of $\hat W(e_k)$ be denoted as $\gamma$, and let the first coordinate of $\hat W(e_t)$ be denoted as $\phi$.
 The first coordinate of $(\hat W(e_i) - \hat W(e_j))  - (\hat W(e_k) - \hat W(e_t))$ is equal to $(\alpha - \beta) -(\gamma -\phi)) =  \phi + \alpha - \beta - \gamma $. Assume that $ \phi + \alpha - \beta - \gamma \geq 0$. We will show that in this case the greatest of these four values is either $\phi$ or $\alpha$. Assume that this is not the case. Without any loss in generality, let $\beta$ be the greatest of these four values. Then $\beta \geq 3 \alpha$ and $\beta \geq 3\phi$. Therefore, 
 $2\beta \geq 3(\phi + \alpha)$, and so $\beta \geq 1.5(\phi + \alpha) > \phi + \alpha$, and $  \phi + \alpha - \beta - \gamma < 0$, which is a contradiction to the original assumption. Next, without any loss in generality, let $\phi$ be the greatest of these four values.
 Since $  \phi \geq  3\beta, \phi \geq 3 \gamma$,   $\phi + \alpha - \beta - \gamma > \phi - ( \beta + \gamma) >0$, 
and since $p> 2\cdot \phi$ and $\phi \geq 3\alpha$, $\phi +\alpha - \beta -\gamma < \phi + \alpha + \beta + \gamma \leq 2\phi < p$. Therefore,
the first coordinate of $(\hat W(e_i) - \hat W(e_j))  - (\hat W(e_k) - \hat W(e_t))$ is not equal to $0\text{ (mod p)}$, and $(\hat W(e_i) - \hat W(e_j))  - (\hat W(e_k) - \hat W(e_t))$ does not belong to $span(\hat F(u))$. 
The case, where $ \phi + \alpha - \beta - \gamma \leq 0$, is resolved in a similar way.

Because of the conditions on step 1 and 2, for every vertex $v$ that passed these steps, we will correctly recover a single set of edges that are incident to $v$ and belong to the initial message, as well as the signs of their weights in the computation of $A_1$, on step 3.

After the third step, we obtain $A_2$ from $A_1$ by subtracting the weights of the reconstructed edges with the reconstructed signs, and then remove from graph $G$ all vertices $v$ remaining after step 2, along with all their adjacent edges. The $i$-th iteration of the attack consists of the same steps as the first, but is applied to the graph remaining after the $(i-1)$-th iteration and the vector $A_i$, also remaining after the $(i-1)$-th iteration. Note that after step 2 of the $i$-th iteration, a vertex $u_i$ is guaranteed to remain, which, when defining the function $\hat W$, had the smallest index $t$ among the indices of the vertices remaining in the graph after the $(i-1)$-th iteration. For $u_i$ and each edge $e$ remaining after the $(i-1)$-th iteration, the $t$-th coordinate of $\hat W(e)$ is 0 if $e$ is not incident with $u_i$,
and is not 0 if $e$ is incident with $u_i$.

Thus, all the arguments for $u$ also apply to $u_i$, meaning that at each iteration we will delete at least one vertex. If at some point $A_i = (0,0,...,0)$,
then instead of the next iteration, we collect the original message from all the edges recovered during the algorithm's iterations.

Note that at the first iteration step, the dimension of the linear span $(m-1)(m-2)/2$ of vectors of length $m$ is calculated $m$ times, at the second iteration step, no more than $m-1 + (m-1)(m-2) + (m-1)(m-2)(m-3) + (m-1)(m-2)(m-3)(m-4)$ checks are performed to determine whether a vector of length $m$ belongs to the linear span $(m-1)(m-2)/2$ of vectors of length $m$,
and at the third iteration step, no more than $(m-1)(m-2) + 2m +1$ checks are performed to determine whether a vector of length $m$ belongs to the linear span $(m-1)(m-2)/2$ of vectors of length $m$. When moving from one iteration to the next,
at most three vectors of length $m$ are added together, and some graph vertices and their adjacent edges are removed. In the original graph of $m$ vertices, at least one vertex is removed after each iteration, so
the algorithm will have at most $m$ iterations. Thus, the algorithm runs in polynomial time in $m$.

End of the proof.
\end{proof}

Non-zero limited noise is a necessary condition for a defense against such an attack, but not a sufficient one. A sufficient condition is to have $dim(span(F(v))) = k$ for every vertex $v$.

Here is a draft of a possible defense against this attack. 
\begin{enumerate}
\item  Let $\hat w_{i,j,i}$ be elements of rapidly increading sequences $(a_{j,1}, ..., a_{j,m-j}),$ where $\forall i \; 3\cdot a_{j,i-1} \leq a_{j,i}$ and $a_{j,1} \geq 3.$ 
\item For $1\leq t <i <j \leq m$, let $\hat w_{i,j,t} = 1$.
\item To set values for  $\hat w_{i,j,t},  1\leq i< j \leq m, i <t$, we will choose $D = (e_{i_1} , ..., e_{i_k} )$ -- a cycle of length $k$.
\item We complete the definition of $\hat W$ on $D$, such that the dimension of $span(\{\hat W(e_{i_1}), ...,\hat W(e_{i_k} )\})$ is equal to $k$.
\item For every vertex $v$ we choose two edges, $e_1(v)$ and $e_2(v)$, that are not incident to $v$ and are not a part of $D$, and complete the definition of 
$\hat W(e_1(v))$  and $\hat W(e_2(v)),$ such that the dimension of the span of the set $\{\hat W(e_{i_1}), ...,\hat W(e_{i_k} ),\hat W(e_1(v)),\hat W(e_2(v))\} \setminus \{\hat W(e)|e \in
D, e$ is incident to $v\}$ is equal to $k$.
\end{enumerate}
After that we have at least $k(k -7)/2$ edges, whose weight $\hat W$ is not completely defined, with the remaining $\hat w_{i,j,t}$ being an arbitrary noise. There is around  $\Omega (k^3)$ of these undefined arbitrary noise elements $\hat w_{i,j,t}$ that we can define in any way we need.


\begin{thebibliography}{0}

\bibitem{endm}
{J. Edmonds.}
Maximum matching and a polyhedron with 0,1-vertices // Journal
of Research of the National Bureau of Standards Section B Mathematics and
Mathematical Physics. 1965. {\bf69B}, №1-2.125--130.


\bibitem{bol-irm}
{A. I. Bolotnikov, A. A. Irmatov.}
On the Matching Arrangement of a Graph, Improper Weight Function Problem and Its Application // Lobachevskii Journal of Mathematics. 2025. {\bf46}, №~3.1025--1039.



\end{thebibliography}
\end{document}